\documentclass[10pt,a4paper]{article}

\usepackage[utf8]{inputenc}
\usepackage{color}
\usepackage{float}
\usepackage{amsmath}
\usepackage{amssymb}
\usepackage{esint}
\usepackage{hyperref}
\usepackage{varioref}
 \usepackage{amsthm}
\usepackage[margin=1in]{geometry}

\newtheorem{theorem}{Theorem}[section]
\newtheorem{corollary}{Corollary}[section]
\newtheorem{lemma}{Lemma}[section]
\newtheorem{definition}{Definition}
\newtheorem{problem}{Problem}[section]
\newtheorem{remark}{Remark}[section]
\newtheorem{proposition}{Proposition}[section]

\makeatletter

\usepackage{caption}

\usepackage{array}
\usepackage{IEEEtrantools}
\usepackage{pdfpages}
\usepackage{tikz}
\usepackage{algorithm}
\usepackage[noend]{algpseudocode}
\usepackage{algorithmicx}
\let\OldStatex\Statex
\renewcommand{\Statex}[1][3]{%
  \setlength\@tempdima{\algorithmicindent}%
  \OldStatex\hskip\dimexpr#1\@tempdima\relax}

\def\ff{\mathbb{F}}
\def\zz{\mathbb{Z}}
\def\al{\alpha}

\def\ft{\mathbb{F}_2}

\def\fq{\mathbb{F}_q}

\def\fp{\mathbb{F}_p}

\def\tx{\tilde{x}}
\def\iP{P^{-1}}
\def\tA{\tilde{A}}

\def\kmA{\mathcal{A}}
\def\beq{\begin{equation}}
\def\eeq{\end{equation}}

\title{Analysis of Periodic Feedback Shift Registers}
\author{Ramachandran Anantharaman\footnote{Department of Electrical Engineering, Indian Institute of Technology - Bombay. e-mail: ramachandran@ee.iitb.ac.in} ~ and Virendra Sule\footnote{Department of Electrical Engineering, Indian Institute of Technology - Bombay. e-mail: vrs@ee.iitb.ac.in}}

\date{}

\begin{document}

\maketitle
\begin{abstract}
This paper develops methods for analyzing periodic orbits of states of linear feedback shift registers with periodic coefficients and estimating their lengths. These shift registers are among the simplest nonlinear feedback shift registers (FSRs) whose orbit lengths can be determined by feasible computation. In general such a problem for nonlinear FSRs involves infeasible computation. The dynamical systems whose model includes such FSRs are termed as Periodic Finite State systems (PFSS). This paper advances theory of such dynamical systems. Due to the finite field valued coefficients, the theory of such systems turns out to be radically different from that of linear continuous or discrete time periodic systems with real coefficients well known in literature. A special finite field version of the Floquet theory of such periodic systems is developed and the structure of trajectories of the PFSS is analyzed through that of a shift invariant linear system after Floquet transformation. The concept of extension of a dynamical system is proposed for such systems whenever the equivalent shift invariant system can be obtained over an extension field.


\end{abstract}

Category : cs.SY, math.DS






\section{Introduction}
\label{sec:intro}
In Cryptography and Communication engineering there is a need for generation of long periodic pseudorandom sequences over finite fields which have a large linear complexity or rank. Generation of such sequences is practically feasible using feedback shift registers (FSRs) whose feedback functions are nonlinear. As dynamical systems in discrete time such devices are modeled by nonlinear \emph{finite state systems} (FSS) or nonlinear maps in vector spaces over finite fields. A specialized nonlinear FSS with simple hardware is obtained by composing two linear FSRs, one master and the other slave, in which the states of master FSR decide the coefficients of the feedback function of the slave FSR. Such a composition is shown in figure \ref{fig:PFSR} which we shall call \emph{compositional FSR}. Now if the initial state of master FSR is chosen such that it lies on a periodic orbit in its state space then we have a \emph{periodic finite state discrete time linear system} (PFSS) which is a linear FSS whose coefficients are periodic functions of time with values in the finite field. This system is essentially nonlinear in the combined (product) state space of master and slave FSRs and involves quadratic product terms of their states. However this system has time varying dynamical equation in the state space of slave FSR due to the projection of the product state space onto the state space of slave FSR. Purpose of this paper is to show many results of the structure of solutions of such PFSS systems. The problem of finding orbit lengths of periodic orbits of general nonlinear FSS is known to be both theoretically and computationally hard \cite{MortRei}. This problem is computationally feasible and theoretically well understood only in the case of linear time invariant FSS \cite{Gill}. In this paper we develop analogous results for this problem for the next general class of systems the PFSS which resolves the practical question of computing periodic orbits for the compositional FSR device described above. Such an analysis of orbits of PFSS does not seem to have been much studied as is evident from almost non-existence of past references. 
\subsection{Compositional Feedback Shift Registers}
Dubrova \cite{Dubrova} suggested composition of FSRs by constructing a feedback function where the feedback function has arguments from all the FSRs and affects feedback of each individual FSR. The effect of this function is to join cycles of smaller periods to create a cycle of maximum period. In this paper, we propose two different constructions for periodic FSRs (PFSR) namely the Fibonacci type PFSR and the Galois type PFSR based on Fibonacci and Galois FSR respectively. A Fibonacci type FSR is a FSS with dynamical evolution given by the nonlinear map as in equation (\ref{eq:GenFFSR})
\begin{equation}
\label{eq:GenFFSR}
    \begin{bmatrix} x_1(k+1) \\ x_2(k+1) \\ \vdots \\ x_{n-1}(k+1) \\ x_n(k+1) \end{bmatrix} = \begin{bmatrix}x_2(k) \\ x_3(k) \\ \vdots \\ x_n(k) \\ g(x_1(k),\dots,x_n(k)) \end{bmatrix}
\end{equation}
The function $g$ is called the \emph{feedback function}. When the function is a linear sum of terms (i.e. $ a_{i}x_{i}$), the FSR is called linear FSR (LFSR). As discussed before, we introduce periodic FSRs (PFSRs) which can be constructed in hardware as composition of two FSRs (master and a slave) as follows. The independent FSR (master FSR) has its own state evolution with a time invariant feedback as in equation (\ref{eq:GenFFSR}). The other FSR (slave FSR) has the states of the master as the feedback coefficient. Let $x_i$ and $y_i$ be the states of slave and master FSRs. The dynamical evolution of the combined FSR is 
\begin{equation}
\label{eq:CombinedFFSR}
    \begin{bmatrix}x_1(k+1) \\ \vdots \\ x_n(k+1) \\ y_1(k+1) \\  \vdots \\ y_n(k+1) \end{bmatrix} = \begin{bmatrix}x_2(k) \\ \vdots \\ \sum y_i(k) x_i(k) \\ y_2(k) \\ \vdots \\ g(y_1(k),\dots,y_n(k)) \end{bmatrix}
\end{equation}
Such a construction is graphically shown in figure \ref{fig:PFSR}.
\begin{figure}[ht]
\centering
\begin{tikzpicture}[scale = 0.75]
    \draw[step=1cm] (0,0) grid (6,-1);
    \node at (0.5,-0.5) {$x_n$};
    \node at (4.5,-0.5) {$x_1$};
    \node at (5.5,-0.5) {$x_0$};
    
    \draw (3,-2) ellipse (3cm and 0.5cm);
    \node at (3,-2) {$\sum y_i x_i$};    
    \draw[step=1cm] (0,-3) grid (6,-4);
    \node at (0.5,-3.5) {$y_n$};
    \node at (4.5,-3.5) {$y_1$};
    \node at (5.5,-3.5) {$y_0$};

    \draw (3,-5) ellipse (3cm and 0.5cm);
    \node at (3,-5) {$g(y_1,\dots,y_n) $};
    \draw[->] (0.5,-1.05) to (0.7,-1.6);
    \draw[->] (0.5,-2.95) to (0.7,-2.4);
    \draw[->] (1.5,-1.05) to (1.6,-1.5);
    \draw[->] (1.5,-2.95) to (1.6,-2.5);
    \draw[->] (4.5,-1.05) to (4.4,-1.5);
    \draw[->] (4.5,-2.95) to (4.4,-2.5);
    \draw[->] (5.5,-1.05) to (5.3,-1.6);
    \draw[->] (5.5,-2.95) to (5.3,-2.4);
    \draw (-0.1,-2) -- (-2,-2);
    \draw (-2,-2) -- (-2,-0.5); 
    \draw[->] (-2,-0.5) to (-0.1,-0.5);
    \draw[->] (0.5,-4.05) to (0.7,-4.6);
    \draw[->] (1.5,-4.05) to (1.6,-4.5);
    \draw[->] (4.5,-4.05) to (4.4,-4.5);
    \draw[->] (5.5,-4.05) to (5.3,-4.6);
    \draw (-0.1,-5) -- (-2,-5);
    \draw (-2,-5) -- (-2,-3.5); 
    \draw[->] (-2,-3.5) to (-0.1,-3.5);
    
    \node at (3,0.5) {Slave FSR};
    \node at (3,-6.5){Master FSR};
 
 
\end{tikzpicture}
\caption{Compositional Fibonacci Periodic FSR}
\label{fig:PFSR}
\end{figure}
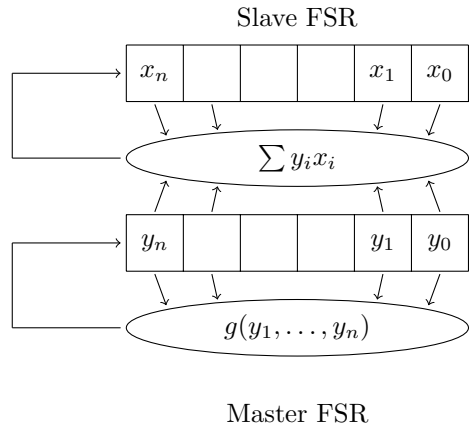
The projection of the dynamics on to the state space of slave FSR is a PFSS with each $A(k)$ matrix in companion form as follows
\begin{equation}
A(k)  = \begin{bmatrix}  0 & 1 & 0 & \dots & 0 \\ 0 & 0 & 1 & \dots & 0 \\ \vdots & \vdots & \vdots & \ddots & \vdots \\ 0 & 0 &0&\dots & 1 \\ y_1(k) & y_2(k) &y_3(k) & \dots & y_n(k) \end{bmatrix}
\label{eq:AkFPFSR}
\end{equation}
When the initial condition of the master FSR lie on a periodic orbit of period $N$, then the dynamics of the slave FSR is defined by a PFSS with $A(i)$ matrices periodic as defined in equation (\ref{eq:AkFPFSR}). 

Similar to the Fibonacci type PFSR, one can construct a Galois type PFSR using a Galois LFSR. Galois LFSR are another class of Linear FSR with the following transition map
\begin{equation}
    \begin{bmatrix} x_1(k+1)\\ x_2(k+1)\\ \vdots  \\ x_n(k+1) \end{bmatrix} = \begin{bmatrix} \al_{1} x_1(k) + x_2(k) \\ \al_2 x_1(k) + x_3(k) \\ \vdots \\ \al_n x_1(k)\end{bmatrix}
\end{equation}
A Galois PFSR is constructed when the coefficients $\al_i$ are periodic and time-varying and comes from another FSR, the master FSR. The combined FSR is given by the following equation
\begin{equation}
\label{eq:CombinedGPFSR}
    \begin{bmatrix}x_1(k+1) \\x_2(k+1) \\ \vdots \\ x_n(k+1) \\ y_1(k+1) \\  \vdots \\ y_n(k+1) \end{bmatrix} = \begin{bmatrix}y_1(k)x_1(k) + x_2(k) \\ y_2(k)x_1(k) + x_3(k) \\ \vdots \\ y_n(k)x_1(k) \\ \al_1 y_1(k) + y_2(k) \\ \vdots \\ \al_n y_1(k) \end{bmatrix}
\end{equation}
A construction for a 3-bit Galois PFSR is given in figure \ref{fig:GPFSR}.
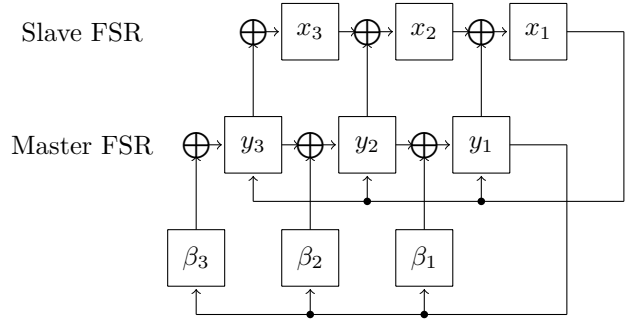
\begin{figure}[ht]
\centering
\begin{tikzpicture}[scale = 0.75]
    \draw[step=1cm] (0,0) grid (1,-1);
    \draw[step=1cm] (2,0) grid (3,-1);
    \draw[step=1cm] (4,0) grid (5,-1);
    \node at (1.5,-.5) {$\bigoplus$};
    \node at (3.5,-.5) {$\bigoplus$};
    \node at (-.5,-.5) {$\bigoplus$};
    \draw[->] (-0.3,-0.5) to (-0.05,-0.5);
    \draw[->] (1,-0.5) to (1.3,-0.5);
    \draw[->] (1.7,-0.5) to (1.95,-0.5);
    \draw[->] (3,-0.5) to (3.3,-0.5);
    \draw[->] (3.7,-0.5) to (3.95,-0.5); 
    
    \node at (0.5,-0.5) {$x_3$};
    \node at (2.5,-0.5) {$x_2$};
    \node at (4.5,-0.5) {$x_1$};
    
    
    \draw[step=1cm] (-1,-2) grid (0,-3);
    \draw[step=1cm] (1,-2) grid (2,-3);
    \draw[step=1cm] (3,-2) grid (4,-3);
    \node at (-1.5,-2.5) {$\bigoplus$};
    \node at (0.5,-2.5) {$\bigoplus$};
    \node at (2.5,-2.5) {$\bigoplus$};
    \draw[->] (-1.3,-2.5) to (-1.05,-2.5);
    \draw[->] (0,-2.5) to (.3,-2.5);
    \draw[->] (0.7,-2.5) to (0.95,-2.5);
    \draw[->] (2,-2.5) to (2.3,-2.5);
    \draw[->] (2.7,-2.5) to (2.95,-2.5); 
    
    \node at (-0.5,-2.5) {$y_3$};
    \node at (1.5,-2.5) {$y_2$};
    \node at (3.5,-2.5) {$y_1$};
    
    
    \draw[step=1cm] (-2,-4) grid (-1,-5);
    \draw[step=1cm] (0,-4) grid (1,-5);
    \draw[step=1cm] (2,-4) grid (3,-5);

    \node at (-1.5,-4.5) {$\beta_3$};
    \node at (.5,-4.5) {$\beta_2$};
    \node at (2.5,-4.5) {$\beta_1$};

    \draw[-] (4,-2.5) to (5,-2.5);
    \draw[-] (5,-2.5) to (5,-5.5);
    \draw[-] (5,-5.5) to (-1.5,-5.5);
    \draw[->] (2.5,-5.5) to (2.5,-5.05);
    \draw[->] (0.5,-5.5) to (0.5,-5.05);
    \draw[->] (-1.5,-5.5) to (-1.5,-5.05);
    \draw[->] (2.5,-4) to (2.5,-2.7);
    \draw[->] (0.5,-4) to (0.5,-2.7);
    \draw[->] (-1.5,-4) to (-1.5,-2.7);
    \draw[-] (5,-0.5) to (6,-0.5);
    \draw[-] (6,-0.5) to (6,-3.5);
    \draw[-] (6,-3.5) to (-0.5,-3.5);
    \draw[->] (3.5,-3.5) to (3.5,-3.05);
    \draw[->] (1.5,-3.5) to (1.5,-3.05);
    \draw[->] (-.5,-3.5) to (-.5,-3.05);
    \draw[->] (3.5,-2) to (3.5,-0.7);
    \draw[->] (1.5,-2) to (1.5,-0.7);
    \draw[->] (-0.5,-2) to (-0.5,-0.7);
    
    \node at (2.5,-5.5) [circle,fill,inner sep=1pt]{};
    \node at (0.5,-5.5) [circle,fill,inner sep=1pt]{};
    \node at (3.5,-3.5) [circle,fill,inner sep=1pt]{};
    \node at (1.5,-3.5) [circle,fill,inner sep=1pt]{};
         \node at (-3.5,-0.5) {Slave FSR};
         \node at (-3.5,-2.5) {Master FSR};
\end{tikzpicture}
\caption{3-bit Galois Periodic FSR}
\label{fig:GPFSR}
\end{figure}

and the matrix $A(k)$ for the slave dynamics is given as 
\begin{equation}
    A(k) = \begin{bmatrix}  y_1(k) & 1 & 0 & \dots & 0 \\ y_2(k) & 0 & 1 & \dots & 0 \\ \vdots & \vdots & \vdots & \ddots & \vdots \\ 0 & 0 &0&\dots & 1 \\ y_n(k) &0  &0 & \dots & 1 \end{bmatrix}
    \label{eq:AkGPFSR}
\end{equation}
So in general when the initial condition of the master FSR lies on a periodic orbit, the periodic FSR is a linear dynamical system of PFSS type
\begin{align}
    \begin{aligned}
    x(k+1) &= A(k) x(k)\\
    A(k+N) &= A(k)
    \end{aligned}
\end{align}
where $A(k)$ matrix is defined as in equation (\ref{eq:AkFPFSR}) or (\ref{eq:AkGPFSR}) depending on the configuration of the PFSR used.

\subsection{Background of FSS}
 Finite state discrete time (dynamical) shift invariant systems called in brief as \emph{Finite State Systems} (FSS) are a class of discrete time dynamical systems which evolve over a finite set where the next state transition depends on an evolution rule which is independent of shift in time. Choosing a sufficiently large non-negative numbers $m,n$ and a finite field $\fp$ for a suitably small prime $p$ these systems can be studied as discrete time dynamical systems with state space $V=\fq^n$ where $q=p^{m}$. The dynamical evolution of the states $x$ in $V$ of the FSS is then given by
\begin{equation}
    x(k+1) = F(x(k))
    \label{eq:FSS}
\end{equation}
where $F : V \rightarrow V$ is called the \emph{state transition map}. $F$ is described by polynomial functions of $n$ variables over $\fq$. We shall denote $\ff=\fq$ as a general finite field to describe FSS Dynamics. FSS has been a subject of interest since a long time in Electrical Engineering, Computer Science and Systems Biology as large number of physically existing systems such as sequential digital logical systems, finite state machines and Biological systems such Genetic networks to just name a few, are modeled by FSS. Several papers and books in the past have studied the topic of FSS \cite{Golomb,Gill2,Goresky,MortRei}. When the map $F(.)$ in (\ref{eq:FSS}) is $\ff$-linear and shift invariant it is denoted by a matrix $A$ over $\ff$. The system is then called \emph{linear} FSS (LFSS). Such a system has a representation as
\beq\label{LFSS}
x(k+1)=Ax(k)
\eeq
where $A$ is a square matrix over $\ff$. When $F(.)$ is an $\ff$-linear map such as linear polynomial on the extension field with $q=p^{m}$ the matrix $A$ shall depend on a choice of a basis. However as the invariant polynomials of $A$ are invariant under change of basis, the dynamical characteristics of the trajectories of the system also remain invariant with respect to a basis. LFSS are also known as \emph{linear modular systems} and have been extensively studied and the structure of their solutions is well understood since long \cite{Gill}.

In this paper we analyze the structure of solutions of the time varying linear FSS with periodic time variation of its coefficients. Such a system is described by the state space equation
\begin{equation}
\label{PFSS}
    x(k+1)=A(k)x(k)
\end{equation}
where sequence of matrices $A(k)$ over $\ff$ are periodic of period $N$, $A(k+N)=A(k)$. Such systems can be naturally called \emph{Periodic Finite State Systems} (PFSS). A striking practical application of such systems is described above by compositional FSR. 

Although FSS have had strong applications, understanding of the orbit parameters of solutions of general nonlinear FSS is still mostly an unresolved problem both theoretically and computationally. On the other hand we shall show in this paper that the special nonlinear FSS arising as PFSS is amenable to deeper analysis under certain conditions. Hence such results are applicable to the realistic systems arising from the interconnection of FSRs described above which are of significance to applications in Cryptography. A more general, time varying dynamical system in finite state space is described by a map $F(k,x):\zz^{+}\times \ff^{n}\rightarrow\ff^{n}$ by equations
\beq\label{TVFSS}
x(k+1)=F(k,x(k))
\eeq
where $x(k)$ denotes the state in $\ff^n$ at an instant $k\geq 0$.

In most of the applications of FSS in Cryptography and Systems Biology, knowledge of lengths of orbits of an FSS is a very important information. Unfortunately, computation of orbits and their lengths in general FSS involves solving NP hard computational problems. Such computations are feasible only in the case of LFSS and the theory behind such computations has been well known \cite{Gill} although algorithms to carry out such computations may still offer scope of further research. In this paper we answer this problem of computation of orbit lengths for PFSS under special conditions.

\subsection{Formulation of problems of PFSS using Floquet transformations}
For a special condition to be followed by PFSS, it is immediately tempting to consider the Floquet theory of periodic dynamical systems. While Floquet theory for continuous time dynamical systems has been a celebrated topic in non-linear dynamics and mechanics since long \cite{yakustar}, a discrete time version of Floquet theory to transform a periodic time varying discrete time linear dynamical system into a shift invariant system LFSS, seems to have developed only recently since the paper \cite{VanD}. Hence it is interesting to consider analog of Floquet theory for PFSS. This shall be the main agenda of this paper. It turns out that the PFSS theory does not follow from the discrete time results of periodic systems established in \cite{VanD}. We develop in this paper the specialized theory appropriate for PFSS.

A brief overview of results is as follows. In section \ref{sec:NSPFSS} we develop the non-singularity conditions for a PFSS and then in section \ref{sec:SolutionsPFSS}, solutions of PFSS is discussed without the use of Floquet theory. Section \ref{sec:FloquetPFSS} develops a Floquet like transform for the PFSS and section \ref{sec:FloquetSolutions} we use the Floquet theory to develop more insights on the solutions of the PFSS and algorithms are developed for computing orbit lengths of specific initial conditions, computing all orbits of a non-singular PFSS and compute an initial condition which initiates PFSS for a prescribed orbit length. Lastly in section \ref{sec:NEPFSRs} we analyze two specific numerical examples of PFSRs.

\section{Non-singularity of the dynamical systems and periodic solutions}
\label{sec:NSPFSS}

Our problem of determining the structure of solutions of periodic systems is simplified when every solution of the dynamical system is periodic. Consider a general FSS over a finite field $\ff$ defined by a state transition map $F:\ff^{n}\to \ff^{n}$. We first state condition on systems such as FSS (\ref{eq:FSS}) to have all solutions periodic. A solution $x(k)$, $k>0$ of (\ref{eq:FSS}) is said to be periodic if there is a time $T>0$ such that $x(k)=x(k+T)$ for any $k$. In fact for (\ref{eq:FSS}) solution $x(k)$ is periodic iff $x(0)=X(T)$. The smallest such $T$ is called the \emph{period} or \emph{orbit length} of the periodic solution. We call (\ref{eq:FSS}) \emph{non-singular} if the map $F:\ff^{n}\rightarrow\ff^{n}$ is one to one. Following observation then follows easily.

\begin{proposition}
\label{prop:NSFSS}
A system FSS (\ref{eq:FSS}) has all solutions periodic iff FSS is non-singular
\end{proposition}

The proof follows easily by considering the fact that a one to one map $F$ on a finite set of states is a permutation of the set which is allways representable by independent cycles.

If a state $x$ lies on a unique periodic orbit then the period coincides with the number of points on the orbit. It follows from above proposition that an LFSS (\ref{LFSS}) is non-singular iff $A$ is a non-singular matrix or $\ker L=0$ where $L$ is the linear map represented by $A$.

\subsection{Non-singularity of periodic TVFSS}
Next we want to look for an analogous non-singularity property of PFSS under which all solutions of PFSS will be periodic. Such a condition can in fact be developed for a more general TVFSS (\ref{TVFSS}) whose state transition map is periodic in time. A TVFSS (\ref{TVFSS}) is said to be \emph{non-singular} if for each $k$, the transition map $F(k,x):\ff^{n}\rightarrow\ff^{n}$ is one to one. Periodicity of all solutions then follows under the additional restriction that $F(k,x)$ is periodic in $k$, i.e. there exists an $N>0$ such that $F(k+N,x)=F(k,x)$ for any $k$. We shall call such a TVFSS (\ref{TVFSS}) as \emph{periodic} TVFSS of period $N$. We then have the analogous result in the next proposition. 

\begin{proposition}
If a TVFSS (\ref{TVFSS}) is periodic then all its solutions are periodic iff the system is non-singular. 
\end{proposition}
 \begin{proof}
 As the map $F(k,x)$ is periodic in $k$ let the period be $N$. Consider the transition map over one period $[0,N-1]$ given by the composition
 \[
 \Phi(N,x)=F(N-1,F(N-2,(\ldots,F(1,F(0,x))\ldots)
 \]
 This is called the \emph{monodromy map} of the system. Since $F(,)$ is non singular, the monodromy map is one to one. Moreover since $F(,)$ is also periodic, the monodromy map defines a dynamical system of the type FSS (\ref{eq:FSS}) as its transition map
 \beq\label{oneperiod}
 y(r+1)=\Phi(N,y(r))
 \eeq
 over $\ff^{n}$. Since the transition map of this FSS is one to one from proposition \ref{prop:NSFSS}, this system has all solutions periodic. Let $x$ be any point in $\ff^{n}$ and $y(0)=x$ then there is a period $L$ of the periodic orbit through this initial point under the dynamics of (\ref{oneperiod} and thus $y(L)=y(0)=x$. Since $\Phi(N,.)$ is a transition map of the original TVFSS we have $x(N)=\Phi(N,x(0))$ and in general
 \[
 x((k+1)N)=\Phi(N,x(kN))
 \]
 Considering $r=kN$ and $y(r)=x(kN)$ the periodic orbit of $y(r)$ of period $L$ gives $x(kN)=x(kN+LN)$ for all $k$. Hence $x(0),\ldots,x(LN)$ is a periodic orbit of the TVFSS (\ref{TVFSS}).
 
 Conversely let all solutions of TVFSS be periodic, then considering the one period transition relation (\ref{oneperiod}) an initial state $y(0)$ is on a unique periodic orbit. Hence $\Phi(N,x)$ is one to one in $x$. Since $\Phi(n,x)$ is the composition of transition maps $F(k,x)$ from $k=0$ to $N-1$ this shows that the each of these composition maps are one to one. Hence due to periodicity of the transition map the TVFSS is nonsingular.  

 \end{proof}
 
Coming to PFSS, the \emph{monodromy matrix} of period $N$, is
\[
\Phi(N) = A(N-1)A(N-2)\cdots A(1)A(0)
\]
Let the state transition map from $k_0$ to $k_1$ be denoted as   
\[
S(k_1,k_0) = A(k_1-1)A(k_1-2)\cdots A(k_0+1)A(k_0)
\]
which gives the transition
\[
x(k_{1})=S(k_{1},k_{0})x(k_{0})
\]
hence $\Phi(N)=S(N,0)$. The periodicity result of TVFSS now gives us,
\begin{proposition}
A PFSS (\ref{PFSS}) of period $N$ has all solutions periodic iff all the system matrices $A(k)$ for $k=0,1,\ldots,N-1$ are non-singular.
\end{proposition}
\begin{proof}
By the previous general result it follows that all solutions of PFSS are periodic iff $\Phi(N)$ is nonsingular. Since this is product of all the system matrices $A(k)$ the result follows.  
\end{proof}

The above proposition characterizes PFSS all of whose solutions are periodic. Under this simplified condition on structure of solutions, we determine orbit lengths of all orbits.  

\section{Orbit lengths of solutions of non-singular PFSS}
\label{sec:SolutionsPFSS}
Previous section showed that a non-singular PFSS has all solutions periodic. The next problem then is to determine the lengths of orbits of different possible periodic solutions of (\ref{PFSS}) for different initial conditions.

Consider a non-singular PFSS (\ref{PFSS}) over $\fq^n$ and the subspace $\kmA \subset \fq^n$ defined as follows 
\begin{equation}
\label{eq:defKMA}
    \kmA = \bigcap_{\substack{i = 0\\ j > i}}^{N-2} ker(A(i) - A(j))
\end{equation}
This subspace allows some classification of orbit lengths as shown in the following theorem. 

\begin{theorem}
\label{thm:CoprimeT}
Let $N$ be the period of the non-singular PFSS. If $x(0)$ has an orbit which is not completely contained in $\kmA$ and has length $T$ then $\gcd(T,N)\neq 1$. 
\end{theorem}
\begin{proof}
Let $T$ be the period of $x(0)$. 
\[
x(0) = x(T) = x(2T) = \ldots
\]
and also
\[
x(1) = x(T+1) = x(2T+1) = \ldots
\]
Let $T = m_1N+r_1$. Since $A(k)$ is of period $N$, $A(T) = A(m_1N+r_1) = A(r_1)$
\begin{align*}
x(T+1) &= A(T) x(T) \\ 
&= A(r_1) x(T) \\
&= A(r_1) x(0)
\end{align*}
But $x(T+1) = x(1)$. So, we have 
\begin{align}
&A(0)x(0) = A(r_1)x(0) \nonumber \\
\implies & x(0) \in ker(A(0)-A(r_1))
\label{eq:lemCoprimeT1}
\end{align}
Similarly if $2T = m_2N+r_2$, with same arguments and the fact that $x(T+1) = x(2T+1)$, it follows from equation (\ref{eq:lemCoprimeT1}) that
\begin{align*}
&x(0) \in ker(A(0)-A(r_2))\\
&x(0) \in ker(A(r_1)-A(r_2))
\end{align*}
Let $T$ be coprime to $N$, the reminders of $T,2T,3T,\dots,(N-1)T$ when divided by $N$ span the entire set $\{1,2,\dots,N-1\}$. 
So, we can write similar equations for each $mT$, $(m < N)$. This will give rise to the condition that 
\[
x(0) \in ker(A(i)-A(j))\ \ \forall\ \ i \neq j
\]
This implies that $x(0) \in \kmA$. For each time instant $k$, the state $x(k)$ has a period $T$. Using the same arguments, we can prove that the whole of $x(k) \in \kmA$ when $\mbox{gcd}(T,N) = 1$. This proves that when an orbit $x(k)$ has atleast one point outside of $\kmA$, then $\mbox{gcd}(N,T) \neq 1$.  
\end{proof}

Of simplest structure are orbits which are fixed points $x$ which satisfy $x=x(0)=x(k)$ for all $k>0$. Hence these are orbits of length $1$. Hence the above theorem implies
\begin{corollary}
\label{cor:FixPtKMA}
All fixed points $x$ belong to $\kmA$.
\end{corollary}

\begin{corollary}
\label{cor:NprimeGen}
Given a non-singular PFSS with period $N$, if $N$ is prime then any initial condition $x(0)$ which does not lie in $\kmA$ has a period which is a multiple of $N$.
\end{corollary}

Subspace $\kmA$ is an intersection of various null spaces defined by matrices $A(k)$. We can thus expect $\kmA$ of small dimension. As a consequence of corollary \ref{cor:NprimeGen}, if $N$ is prime we can expect that apart from a small set of initial conditions, all other initial conditions have a period which is a multiple of $N$. Hence we also have

\begin{corollary}
If the period $N$ is prime and $\kmA=0$ then all orbits of non-singular PFSS have orbits which are multiples of $N$ except for $x(0) = 0$.
\end{corollary}

So far above results on non-singular PFSS did not employ any Floquet transform. While much of the theory of linear periodic dynamical systems has existence of Floquet transform this was not so even for the discrete time systems studied in \cite{VanD}. Over finite fields existence of Floquet transform is still more special. In the next section we establish further results for these special systems.

\section{Floquet Theory for PFSS}
\label{sec:FloquetPFSS}
Floquet transform has been an important tool to study periodic linear dynmical systems and is a time varying transformation of the state which converts the periodic linear system to a time invariant system. (Floquet transform is also known and is same as the Lyapunov transform in \cite{Gantmacher} for periodic differential equations). The Floquet transform approach for discrete time periodic systems is developed in \cite{VanD}. While the transform exists for continuous periodic systems under very mild conditions, in the discrete time case (with real valued states) the transform exists under certain conditions on the system matrices $A(k)$ (as proved in \cite{VanD}). We develop an analogous result for PFSS which have finite field valued and finite states. The Floquet transform we seek is a periodic non-singular linear transformation $P(k)$ of the state space of the same period $N$ as the PFSS which transforms state to
\[
\tilde{x}(k)=P(k)x(k)
\]
then the dynamics of the transformed state is
\begin{align}\label{FT}
\tilde{x}(k+1)&=P(k+1)x(k+1) \nonumber \\
&=P(k+1)A(k)P(k)^{-1}\tilde{x}(k)
\end{align}
The objective of the Floquet tranform is to find $P(k)$ such that the system matrix 
\begin{align}
\label{FT1}
\tilde{A}(k)=P(k+1)A(k)P(k)^{-1}
\end{align}
of the transformed state is shift invariant $\tilde{A}$. The PFSS then transforms to the LFSS
\[
\tx(k+1)=\tA\tx(k)
\]
Hence the structure of orbits of PFSS can now be studied completely from those of this shift invariant system which has a well developed theory.

\begin{theorem}
\label{thm:FLT1}
Given a non-singular PFSS with period $N$, the Floquet transform exists iff the monodromy matrix has an N-th root.
\end{theorem}
\begin{proof}
Given a PFSS let the Floquet transform be denoted by $P(k)$. If there exists a solution $P(k)$ and $\tA$ for the system of equations (\ref{FT1}), evaluating these equations for each $k$, we get 
\begin{align}
    \begin{array}{rcl}
     \tA  &=& P(1)A(0) \iP(0) \\
     \tA  &=& P(2)A(1) \iP(1) \\
    \ \ &\vdots& \\
     \tA  &=& P(0) A(N-1) \iP(N-1)
    \end{array}
\label{eq:thmFLT1}
\end{align}
Hence multiplying all the equations 
\begin{align}
\tA^N &= P(0) A(N-1)\cdots A(1) A(0) \iP(0) \nonumber \\
&= P(0) \Phi \iP(0)
\label{eq:thmFLT2}
\end{align}
Assuming $P(0)=I$ without loss of generality, a necessary condition for Floquet transform to exist is the existence of an $N$-th root of the monodromy matrix.
\[
\tA^N = \Phi
\]
The transformation $P(k)$ is now solved in terms of $\tA$ as
\[
P(i)=\tA P(i-1) A^{-1}(i-1)
\]
for $i=1,2,\ldots,N-1$ with $P(0)=I$. 

The sufficiency condition stems from the fact that if the $N$-th root of the monodromy matrix exists and since each $A(k)$ is non-singular, we can compute each $P(k)$ using the equations (\ref{eq:thmFLT1}).   
\end{proof}

When Floquet transform exists, it can be assumed without loss of generality that $P(0) = I$ and compute the solutions of other $P(k)$ since any other solution to $P(k)$ with different $P(0)$ will be related through the similarity transform $P(0)$. 

\begin{remark}
A question arises in which field do we consider an $N$-th root $\tA$ of the monodromy matrix? We shall choose to seek an $N$-th root in any extension field of $\ff$. This makes the finite field Floquet theory of PFSS to extend the possibility of the resulting LFSS system to be over more general extended finite field analogous to computing Floquet transform over complex numbers in the real state space systems. 
\end{remark}

\subsection{Extension of PFSS and equivalent LFSS}
From the analysis of the existence of Floquet transform for PFSS in the extension field it is now logical to extend the theory of PFSS to the smallest possible field extension over which all possible $N$-th roots of the monodromy matrix exist. Hence it is necessary to redefine the field of definition of the PFSS and the equivalent LFSS.

\begin{definition}
Consider a non-singular PFSS defined in (\ref{PFSS}) over a field $\fq$ of period $N$ and let $\Phi(N)$ denote its monodromy matrix. If $\ff_{q^m}$ is the smallest extension of $\ff_q$ such that all possible $N$-th roots of $\Phi(N)$ exist over $\ff_{q^{m}}$, then for any one such $N$-the root $\tA$ over $\ff_{q^m}$ and the corresponding solution $P(k)$ of the Floquet transform obtained from equations (\ref{eq:thmFLT1}), the PFSS defined over the extended state space $\ff_{q^{m}}^n$ and defined by the same equation (\ref{PFSS}) is called the \textbf{Extended PFSS} and the LFSS system
\beq\label{eqLFSS}
\tx(k+1)=\tA\tx(k)
\eeq
with the state space extended to $\ff_{q^{m}}^n$ is called the \textbf{Equivalent LFSS} of the extended PFSS. 
\end{definition}

Hereafter whenever a PFSS is considered with existence of a Floquet transform we shall mean the extended PFSS and its equivalent LFSS over the extended state space. This kind of extended PFSS allows a systematic analysis of the orbit structure of the PFSS in a sense similar to how algebraically closed fields allow systematic investigation of various results of roots of polynomials. In systems theory such an extension of the original system is a vast generalization of the notion of complexification of dynamical systems. In real systems (defined over real coefficients and real state space) complexification is not always realistic because of complex coefficients. However any finite state system over any finite field can be synthesized by appropriate digital electronic hardware. Hence this theory of extension of PFSS is realistic in systems theory.  

\subsection{Floquet transform for real systems and example of failure of Floquet transform over finite fields}
Over finite fields $\fq$ a matrix may or may not have an $N$th root even in an extension field. This is the reason that the Floquet theory of PFSS does not follow from that of the discrete time Floquet theory of Van Dooran and Sreedhar \cite{VanD}. For completeness we present below the result of Van Dooran and Sreedhar \cite{VanD} and also present an example of a PFSS for which Floquet transform does not exist even over extension fields.

\begin{theorem}[Van Dooran and Sreedhar (\cite{VanD})]
\label{th:VDandJS}
A Floquet transform exists for a periodic linear system iff the following rank conditions hold.
\[
rank(A_{j+i-1}\dots A_{j+1}A_{j}) = r_i \ \ independent\ of \ j
\]
for $1\leq i \leq n$ and $0 \leq j \leq N-1$ 
\end{theorem}

In this theorem the system matrices $A(k)$ are real matrices and the state space is a real cartesian space. The following PFSS system over a finite field however serves as a counter example. Consider a PFSS with period $N = 2$ over $\ft^2$. The $A(k)$ matrices are given below as
\begin{align*}
        A(0) = \begin{bmatrix} 1&1 \\ 1 & 0\end{bmatrix} &\hspace{1cm}
        A(1) = \begin{bmatrix} 1&0\\1 & 1\end{bmatrix}
\end{align*}
which satisfies the rank condition of theorem \ref{th:VDandJS}. The monodromy matrix is computed as 
\[
\Phi  = A(1)A(0) = \begin{bmatrix} 1&1 \\ 0 & 1\end{bmatrix}
\]
From theorem \ref{th:VDandJS}, the above system should have a Floquet transform. But the monodromy matrix $\Phi$ does not have a square root over any field of characteristic $2$.

\subsection{N-th root of a matrix}
Finding the roots of a matrix over finite field is a computationally hard problem. Over real field matrix roots are well known and always exist for nonsingular matrices \cite{Gantmacher}. Over finite fields \cite{Otero} gives necessary and sufficient conditions for the N-th root to exist in the same field. The matrix $N$-th root problem over finite fields is formally stated as follows.
\begin{problem}
    Given a matrix $A \in \ff_q^{n\times n}$ and a non-negative integer $N$, does there exist a matrix $B$ in $\ff_{q^m}^{n\times n}$ for some $m \geq 0$ such that 
    \[
    B^N = A
    \]
\end{problem}
This problem is slightly different from that in \cite{Otero} because we look for existence of roots of the matrix over any extension of the field over which the matrix $A$ is defined. 

From a computational point of view this problem is equivalent to solving a system of $n^2$ polynomial equations in $n^2$ unknowns over a finite field which is of NP-hard class and its further analysis is beyond the scope of this paper. For this paper we assume that the $N$-th root of the monodromy matrix (if it exists) can be computed in feasible time. Once the equivalent LFSS is constructed, we can refine the theory developed in section \ref{sec:SolutionsPFSS} using the additional information gained through the equivalent LFSS.


\section{Solutions of non-singular PFSS with Floquet transform}
\label{sec:FloquetSolutions}
In section \ref{sec:SolutionsPFSS} we partitioned the state space into two sets, a subspace $\kmA$ and its complement $\kmA^{c}$ in the state space $\ff^{n}$ and discussed about the orbit lengths of initial conditions based on these sets. Under the assumption that $N$ is prime, the corollary \ref{cor:NprimeGen} proves that any initial condition $x(0)$ not inside $\kmA$ has an orbit length which is a multiple of $N$. In this section, we further assume that the PFSS allows a Floquet transform and there exists an equivalent LFSS. With this assumption we can explicitly compute the orbit length of the PFSS. First we start with an observation between the non-singular PFSS and the equivalent LFSS.
\begin{lemma}
If the PFSS is non-singular, then the equivalent LFSS is also non-singular.
\end{lemma}
\begin{proof}
If the PFSS is non-singular, then the monodromy matrix is $\Phi$ is also non-singular. This means, the $N$-th root of the monodromy matrix is also non-singular.   
\end{proof}
The above lemma guarantees that for a non-singular PFSS, all solutions of the equivalent LFSS are also periodic orbits. Hence the structure of solutions of PFSS and the equivalent LFSS are same when the PFSS is non-singular. Given a PFSS with Floquet transform and an initial condition $x(0)$ of the PFSS, define an \emph{equivalent initial condition} of the LFSS as
\beq
\tx(0) = P(0) x(0)
\eeq
The period for any initial condition of the equivalent LFSS can be computed using the linear theory which is well known from \cite{Gill}
\begin{theorem}
\label{thm:SolnFloqGenN}
Given a non-singular PFSS with period $N$ and an initial condition $x(0)$. If the equivalent initial condition $\tx(0)$ has a period $T$ under the dynamics of equivalent LFSS, then the initial condition $x(0)$ of the PFSS has a period which divides $\mbox{lcm}(T,N)$
\end{theorem}
\begin{proof}
Given that the initial condition $\tx(0)$ of the equivalent LFSS has a period $T$
\[
\tx(T) = \tx(0)
\]
Since $T$ divides $\mbox{lcm}(T,N)$, let $\mbox{lcm}(T,N) := lT$ for some $l$, 
\begin{align*}
\begin{array}{rrll}
    &\tx(k+lT) &= \tx(k) & \forall \ k \\
    \implies& P(k+lT) x(k+lT) &= P(k) x(k)& \forall \ k \\
    \implies& x(k+lT) &= x(k) & \forall \ k 
\end{array}
\end{align*}
The last equation comes from the fact that $P$ is periodic of period $N$ and $N$ divides $\mbox{lcm}(T,N)$. Also, since $x(k+lT) = x(k)$ for all $k$, the period of $x(0)$ divides $\mbox{lcm}(T,N)$  
\end{proof}
The above theorem provides insight on the possible orbit lengths for the initial condition $x(0)$ of the PFSS through the equivalent initial condition $\tx(0)$.
\subsection{PFSS with prime period}
Further results in this paper assume that the PFSS has a prime period. We have the following results which relates the solutions for a specific initial condition of the PFSS and its equivalent LFSS.
\begin{lemma}
Given a non-singular PFSS with a prime period $N$, if an initial condition $x(0)$ is a fixed point, then the equivalent initial condition $\tx(0)$ is either a fixed point or has an orbit of length $N$
\end{lemma}
\begin{proof}
Since $x(0)$ is a fixed point, $x(i) = x(0)\ \forall\ i \geq 0$. In particular, 
\begin{align*}
\begin{array}{rrcl}
    &x(N) &=& x(0) \\ 
    \implies & \iP(N) \tx(N) &=& \iP(0) \tx(0) \\
    \implies &\tx(N) &=& \tx(0)
\end{array}
\end{align*}
So, $\tx(0)$ has a period which divides $N$. Since $N$ is prime, the period of $\tx(0)$ is either $1$ or $N$  
\end{proof}

\begin{lemma}
\label{lem:ELFSSfixedpoint}
 Given a non-singular PFSS with prime period. If an initial condition $\tx(0)$ of the equivalent LFSS is a fixed point, then $x(0) = \iP(0)\tx(0)$ is either a fixed point or has an orbit of length $N$ under the dynamics of the PFSS. 
 \end{lemma}
\begin{proof}
Given that the initial condition $\tx(0)$ of the equivalent LFSS is a fixed point. So, $\tx(i) = \tx(0)$ for all $i \geq 1$. In particular,
\begin{align*}
\begin{array}{rrcl}
    &\tx(k+N) &=& \tx(k) \ \ \forall\ k\\ 
    \implies & P(k+N) x(k+N) &=& P(k) x(k) \ \ \forall\ k \\
    \implies & x(k+N) &=& x(k) \ \ \forall\ k
\end{array}
\end{align*}
This means, the initial condition $x(0)$ has a period which divides $N$ and since $N$ is prime, we get that $x(0)$ has a period which is either 1 or N.  
\end{proof}
\begin{lemma}
\label{lem:OrbitMult}
Consider a non-singular PFSS of prime period $N$. Let $x(0)$ not belonging to $\kmA$ be an initial condition of the PFSS and $\tx(0)$ be the equivalent initial condition of the equivalent LFSS. Then the orbit length of $\tx(0)$ under the dynamics of the equivalent LFSS is a divisor of the orbit length of $x(0)$ under the dynamics of the PFSS.
\end{lemma}
\begin{proof}
Given an initial condition $x(0)$ of the non-singular PFSS, let the orbit of $x(0)$ be of length $l$. Since $x(0)\notin \kmA$, from \ref{cor:NprimeGen}, we know that $l = m_1 N$ for some $m_1 \geq 1$. So, 
\begin{align*}
\begin{array}{rrcl}
    &x(m_1N) &=& x(0) \\
    \implies & \iP(m_1N) \tx(m_1N) &=& \iP(0)\tx(0) \\
    \implies & \iP(0) \tx(m_1N) &=& \iP(0)\tx(0) \\
    \implies & \tx(m_1N) &=& \tx(0)
\end{array}
\end{align*}
If $\tx(0)$ has a period $T$, then from the above equation $T$ divides $l = m_1N$  
\end{proof}

\begin{theorem}
\label{thm:orbitLCM}
Given a non-singular PFSS with a prime period $N$ and let $x(0)$ be an initial condition of PFSS not belonging to $\kmA$. If the initial condition $\tx(0) = P(0) x(0)$ has a orbit of length $T$ under the equivalent LFSS, then the initial condition $x(0)$ has an orbit of length $\mbox{lcm}\,(T,N)$ under the PFSS.
\end{theorem}
\begin{proof}
Since $x(0) \notin \kmA$, from corollary \ref{cor:NprimeGen}, let the orbit length of $x(0)$ under the PFSS be $m_1N$. Let the initial condition $\tx(0)$ have an orbit of length $T$. From lemma \ref{lem:OrbitMult} we know that $T$ divides $m_1N$. Assuming $l = \mbox{lcm}(T,N)$, from \ref{thm:SolnFloqGenN} the period of $x(0)$ of the PFSS divides the lcm(T,N). We now prove that when $N$ is prime and $x(0) \notin \kmA$ the period is actually the $\mbox{lcm}(T,N)$. Since $N$ is prime, we divide it to three cases of $T$ where $T$ is a multiple of $N$, $N$ is a multiple of $T$ and $T$ is coprime to $N$. \\
\textbf{Case 1:} \emph{$T$ is a multiple of N and $\mbox{lcm}(N,T) = T$}  
\begin{align*}
    & \tx(k+T) = \tx(k) \\
    \implies & x(k+T) = x(k) 
\end{align*}
Consider any other $T_1 < T$ and a multiple of $N$. If $x(k+T_1) = x(k)$, then $T_1$ is a period of the equivalent LFSS too. This is because
\begin{align*}
\begin{array}{rrcl}
    &x(k+T_1) &=& x(k) \\
    \implies &\iP(k+T_1) \tx(k+T_1) &=&\iP(k) \tx(k)  \\ 
    \implies& \tx(k+T_1) &=& \tx(k) \ \ (\mbox{as}\ N\  \mbox{divides}\  T_1)
\end{array}
\end{align*}
This is a contradiction that $T$ is a period of $\tx(0)$ in the equivalent LFSS. So, the period of $x(0)$ is $T$ which is also the $\mbox{lcm}(N,T)$.

\textbf{Case 2:} \emph{$N$ is a multiple of $T$ and $\mbox{lcm}(N,T) = N$}\\
Since $N$ is prime, the only multiple of $N$ are $1$ and $N$. $T=N$ is discussed in the previous case. When $T=1$, we know from lemma \ref{lem:ELFSSfixedpoint} that the initial condition of $x(0)$ of the PFSS is either a fixed point or an orbit of length $N$. But if $x(0)$ is a fixed point then $x(0) \in \kmA$ which is a contradiction to the assumption $x(0) \notin \kmA$. So $x(0)$ has a period $N$ which is the $\mbox{lcm}(T,N)$. 

\textbf{Case 3:} \emph{$T$ and $N$ are co-prime and 
$\mbox{lcm(T,N)} = TN$}\\
 Since $T$ and $N$ are co-prime, $\gcd(T,N) = 1$ and $\mbox{lcm}\,(T,N) = TN$. From theorem (\ref{thm:SolnFloqGenN}) we know that the period of $x(0)$ divides the $\mbox{lcm}(T,N) = TN$. Also from corollary \ref{cor:NprimeGen}, we know that the period of $x(0)$ is a multiple of $N$. We know from lemma \ref{lem:OrbitMult}, that period of $x(0)$ is a multiple of $T$. The smallest multiple of both $T$ and $N$ is their lcm which is $TN$ when both are co-prime. 
 
 So, we see that the period of $x(0) \notin \kmA$ under PFSS is the lcm(N,T) where $T$ is the period of $\tx(0)$ in the equivalent LFSS   
 \end{proof}

\begin{corollary}
Given a non-singular PFSS with prime $N$ and $\kmA = \{0\}$. For any non-zero initial condition $x(0)$ of the PFSS, if the initial condition $\tx(0)$ has a period $L$, then the period of $x(0)$ is $\mbox{lcm}\,(N,L)$
\end{corollary}
The above result is important from a design perspective. Suppose one need to design a PFSS such that all non-zero initial conditions need to have orbit lengths which are integral multiples of the period $N$, then one simple way is to construct the matrices $A(k)$ such that $\kmA = \{0\}$. 

\subsection{Algorithms}
In this section we summarize above results by algorithms for various computations on a given PFSS. A non-singular PFSS is specified by the one period of its periodic transition matrices $A(0),A(1),\ldots,A(N-1)$ along with an $N$-th root of its monodromy matrix with the extension field in which the $N$-th root is chosen.  We discuss the following algorithms for a PFSS with prime period having a Floquet transform.
\begin{enumerate}
    \item Given an initial condition $x(0)$, compute the orbit length of $x(0)$.
    \item Generate the set of all possible orbit lengths in the original state space and the extended state space.
    \item Given an orbit length of the PFSS, find an initial condition $x(0)$ which initiates an orbit of that length. 
\end{enumerate}

Algorithm \ref{algo1} helps in computing the orbit length of an initial condition $x(0)$ of the PFSS with prime $N$. This is the compilation of all the theory developed till now.
\begin{algorithm}
  \caption{Orbit length for initial condition with prime $N$ with Floquet transform}
\label{algo1}
\begin{algorithmic}[1]
\Procedure{Computing orbit length for an initial condition $x(0)$ of a PFSS }{}
\State $\Phi = A(N-1)A(N-2)\cdots A(0)$
\State $\kmA = \bigcap_{i \neq j} ker(A(i) - A(j))$ 
    \State $\tA = \Phi^{1/N}$,     $\tx(0) = x(0)$
    \State $T$ = Orbit length of $\tx(0)$ under equivalent LFSS 
    \Statex  $\tx(k+1) = \tA \tx(k)$
    \If{$x(0) \notin \kmA$}
        \State Orbit length of $x(0) = \mbox{lcm}(N,T)$
    \Else
        \State Orbit length of $x(0)$ divides lcm(N,T)
    \EndIf


\EndProcedure
\end{algorithmic}
\end{algorithm}

The next one (algorithm \ref{algo2}) focuses on listing down all the orbits of the non-singular PFSS with prime $N$, with the assumption that $\kmA = \{0\}$ using the solutions of equivalent LFSS. Once the equivalent LFSS is computed, all its solutions are obtained using the theory developed by \cite{Gill}. Given a non-singular LFSS, if there are $n_i$ distinct orbits of length $T_i$, then the cycle set $\Sigma$ of the LFSS is defined as follows
\beq \label{eq:CSet}
\Sigma = \{ n_0 [1] + n_1[T_1]  + \cdots + n_r [T_r] \}
\eeq
where $T_1,T_2,\dots,T_r$ are the different orbit lengths which can be achieved by the LFSS. Since it is an LFSS, the origin is always a fixed point. All fixed points add up to $n_0$. Computation of cycle set is through the periods of elementary divisors of the LFSS. 

Once the equivalent LFSS of the PFSS is computed, all the solutions of the LFSS is computed by finding the elementary divisors of the system matrix $\tA$. Since $\kmA =\{0\}$, then all initial conditions other than $x(0) = 0$ has an orbit which is a multiple of $N$. So, if we list out all orbit lengths $T_i$ of the equivalent LFSS, then all orbits of the PFSS with the exception of the fixed point are of length $NT_i$
\begin{algorithm}

 \caption{All orbit lengths with $N$ prime, $\kmA = \{0\}$ with Floquet transform}
\label{algo2}
\begin{algorithmic}[1]
\Procedure{Listing orbits of PFSS with prime $N$ and $\kmA = \{0\}$ }{}
\State $\Phi = A(N-1)A(N-2)\cdots A(0)$
\State $\tA = \Phi^{1/N}$
\State Compute the cycle set $\Sigma$ of the equivalent LFSS
\State $\Sigma = \{ n_0 [1] + n_1 [T_1] + \dots n_r [T_r] \}$
\State $\Sigma_{PFSS} = \{ 1[1] + (n_0-1) [N] + n_1[NT_1] + \dots $
\Statex \hspace{1cm}$+ n_r[NT_r]\}$
\EndProcedure
\end{algorithmic}
\end{algorithm}

The last one (algorithm \ref{algo3}) helps in finding an initial condition $x(0)$, such that $x(0)$ has an orbit of a prescribed length (from the set of achievable orbit lengths). \cite{Gill} in his work (Chapter 5, section 15) has solved the equivalent problem of finding the initial condition which achieves a prescribed orbit lengths for a non-singular LFSS. The same problem can be posed for the PFSS also and we call upon the algorithm proposed by Gill to solve the initial condition problem for the PFSS. For a non-singular PFSS with prime $N$ and $\kmA = \{0\}$, all orbits other than the fixed point are a multiple of $N$. Also, the orbit length for the same initial condition under the equivalent LFSS and the PFSS differ by a factor $N$. Given any orbit length $N_1$ for which we need to find the initial condition for the PFSS, we first divide $N_1$ by $N$ and then find the initial condition $\tx(0)$ for the orbit length $N_1/N$ in the equivalent LFSS using the linear theory and this initial condition gives an orbit of length $N_1$ under the dynamics of the PFSS
\begin{algorithm}

 \caption{Initial condition for a prescribed orbit length}
\label{algo3}
\begin{algorithmic}[1]
\Procedure{Computation of initial condition for a given orbit length $N_1$}{}
\State $T =  N_1/N$
\State Compute equivalent LFSS : $\tx(k+1) = \tA \tx(k)$
\State Compute the initial condition $\tx(0)$ for the 
\Statex \hspace{-1.1cm} equivalent LFSS such that $\tx(0)$ has \Statex \hspace{-1cm} orbit length $T$
\State $\tx(0)$ has orbit length $N_1$ for the PFSS
\EndProcedure
\end{algorithmic}
\end{algorithm}

Algorithm 3 reduces the problem of computing initial condition of the PFSS of a prescribed orbit length to the problem of computing initial condition of the LFSS of an appropriate orbit length determined by the period $N$.

\section{Numerical Examples}
\label{sec:NEPFSRs}
In this section, we discuss two constructions of PFSRs (Fibonacci type and Galois type) and compute the orbit lengths using the Floquet theory developed for PFSRs. 
\subsection{Fibonacci type PFSR over $\ff_2$}
 Consider a Fibonacci type PFSR as in figure \ref{fig:3bitPFSR}.
 \begin{figure}[ht]
\centering
\begin{tikzpicture}[scale = 0.85]
    \draw[step=1cm] (0,0) grid (3,-1);
    \node at (0.5,-0.5) {$x_2$};
    \node at (1.5,-0.5) {$x_1$};
    \node at (2.5,-0.5) {$x_0$};
    
    \draw[step=1cm] (0,-2) grid (3,-3);
    \node at (0.5,-2.5) {$y_2$};
    \node at (1.5,-2.5) {$y_1$};
    \node at (2.5,-2.5) {$y_0$};

    \draw[step=1cm] (0,-4) grid (3,-5);
    \node at (0.5,-4.5) {$\beta_2$};
    \node at (1.5,-4.5) {$\beta_1$};
    \node at (2.5,-4.5) {$\beta_0$};

    \draw[->] (0.25,-3) to (0.25,-4);
    \draw[->] (1.25,-3) to (1.25,-4);
    \draw[->] (2.25,-3) to (2.25,-4);
    \draw[->] (0.25,-5) to (0.25,-5.3);
    \draw[->] (1.25,-5) to (1.25,-5.3);
    \draw[-] (2.25,-5) to (2.25,-5.5);
    \draw[->] (2.25,-5.5) to (1.45,-5.5);
    \node at (1.25,-5.5) {$\bigoplus$};
    \draw[->] (1.05,-5.5) to (0.45,-5.5);
    \node at (0.25,-5.5) {$\bigoplus$};
    \draw[-] (0.05,-5.5) to (-1,-5.5);
    \draw[-] (-1,-5.5) to (-1,-2.5);
    \draw[->] (-1,-2.5) to (0,-2.5);

    \draw[->] (0.75,-1) to (0.75,-2);
    \draw[->] (1.75,-1) to (1.75,-2);
    \draw[->] (2.75,-1) to (2.75,-2);
    \draw[->] (0.75,-3) to (0.75,-3.3);
    \draw[->] (1.75,-3) to (1.75,-3.3);
    \draw[-] (2.75,-3) to (2.75,-3.5);
    \draw[->] (2.75,-3.5) to (1.95,-3.5);
    \node at (1.75,-3.5) {$\bigoplus$};
    \draw[->] (1.55,-3.5) to (0.95,-3.5);
    \node at (0.75,-3.5) {$\bigoplus$};
    \draw[-] (0.55,-3.5) to (-0.5,-3.5);
    \draw[-] (-0.5,-3.5) to (-0.5,-0.5);
    \draw[->] (-0.5,-0.5) to (0,-0.5);
    
    
     \node at (5,-0.5) {Slave FSR};
     \node at (5,-2.5) {Master FSR};
 
 
\end{tikzpicture}
\caption{Fibonacci $3$-bit Periodic FSR}
\label{fig:3bitPFSR}
\end{figure}
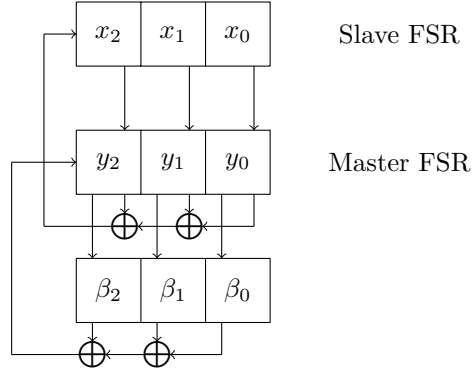
 For the PFSR to be non-singular, the coefficient $y_0(k)$ for the bit $x_0(k)$ is non-zero always. So, the master FSR can be considered as a 2-bit FSR and the slave a 3-bit FSR. Consider the following 2 state FSR for the master configuration. 
 \[
\big[ y_1(k+1),y_2(k+1)\big]^T = \big[ y_2(k), y_1(k)+y_2(k) \big]^T
 \]
 and the following 3 state FSR for the slave configuration.
 \begin{align*}
 \big[&x_0(k+1),x_1(k+1), x_2(k+1) \big]^T = \\
 &\big[ x_1(k),x_2(k), x_0(k) + y_1(k)x_1(k)+y_2(k)x_1(k) \big]^T
 \end{align*}
 gives the following representation for the $A(k)$ matrices.
 \[
 A(k) = \begin{bmatrix} 0 & 1 & 0 \\ 0 & 0 & 1 \\ 1 & y_1(k)  & y_2(k) \end{bmatrix} 
 \]
 Now for an initial condition $[y_1(0),y_2(0)]^T = \big[1,1\big]^T$ of the master, the state evolution is given by 
 \[
 \begin{bmatrix} 1 \\ 1 \end{bmatrix} \to \begin{bmatrix} 1 \\ 0 \end{bmatrix} \to \begin{bmatrix} 0 \\ 1 \end{bmatrix} \to \begin{bmatrix} 1 \\ 1 \end{bmatrix} \to \cdots
 \]
 So, the $A(k)$ matrices are periodic of period $3$ which makes $A(k)$ matrices periodic of period $3$. The matrices are given as follows:
 \begin{align*}
  A(0) = \begin{bmatrix} 0 & 1 & 0 \\ 0 & 0 & 1 \\ 1 & 1 & 1 \end{bmatrix} \hspace{0.2in} A(1) = \begin{bmatrix} 0 & 1 & 0 \\ 0 & 0 & 1 \\ 1 & 1 & 0 \end{bmatrix} \hspace{0.2in}
  A(2) = \begin{bmatrix} 0 & 1 & 0 \\ 0 & 0 & 1 \\ 1 & 0 & 1 \end{bmatrix} 
 \end{align*}
 Note that the matrices are non-singular and the monodromy matrix $\Phi$ is given by
 \[
 \Phi = A(2)A(1)A(0) = \begin{bmatrix} 1&1&1 \\ 0&1&1 \\ 0&1&0\end{bmatrix}
 \]
 Also, the subspace $\kmA$ has the only one non-zero vector
 \[
 \kmA = \mbox{span}\{\begin{bmatrix} 1 \\0 \\ 0 \end{bmatrix}\}
 \]
 So any nonzero initial condition not equal to this vector will have an orbit length which is a multiple of $3$. For Floquet transform to exist, the monodromy matrix should have a cube root. We see that in the extension field $\ff_{2^2}$, $\Phi$ is diagonalizable and 
 \[
 \Phi \sim \begin{bmatrix} 1 & 0 & 0 \\ 0 & \theta & 0 \\ 0 & 0 & \theta^2 \end{bmatrix}
 \]
 where $\theta^2+\theta+1 = 0$. 
 \begin{remark}
 If a matrix $A$ over $\ff$ is diagonalizable over any extension field $\mathbb{K}$, then for any $n$, there exists a suitable field extension $\mathbb{K}_1$ of $\mathbb{K}$ and a matrix $B \in \mathbb{K}_1$ such that 
 \[
 B^n = A
 \]
 \end{remark}
 So, to find the cube root of $\Phi$, one needs to extend the field $\ff_{2^2}$ such that the polynomials $x^3+\theta$ and $x^3+\theta^2$ have roots. So we have the following tower of fields 
 \[
 \ff_2 \xrightarrow{x^2+x+1} \ff_{2^2} \xrightarrow{x^3+\theta} \ff_{2^{6}}
 \]
 Let $\al^3+\theta = 0$. The cube root of $\Phi$ is computed
 \[
 \tA = \Phi^{1/3} = \begin{bmatrix} 1 & \al+\al^2 & 1+\theta^2\al + \theta\al^2 \\ 0 & \theta \al + \theta^2\al^2 & \al + \al^2 \\ 0 & \al+\al^2 & \theta^2\al + \theta \al^2
 \end{bmatrix}
 \]
 Thus a Floquet transform exists for the original PFSR and assuming $P(0) = I$, the state transformation matrices are given as follows:
\[
P(1) = \begin{bmatrix} 
\al^{2} + \al + 1 & \theta \al^{2} +\theta^2 \al & 1 \\
\theta^2 \al^{2} + \theta \al & \al^{2} + \al & 0 \\
\al^{2} + \al & \theta \al^{2} + \theta^2 \al & 0
 \end{bmatrix}  
\hspace{0.2in}
P(2)= \begin{bmatrix}
\theta \al^{2} + \al + 1 & 1 & \al^{2} + \theta \al + 1 \\
\theta^2 \al^{2} + \theta^2 \al & 0 & \theta \al^{2} + \al \\
\theta \al^{2} + \al & 0 & \al^{2} + \theta \al
\end{bmatrix}
\]
and the equivalent LFSS is 
\beq \label{eq:EQLFSS}
\tx(k+1) = \begin{bmatrix} 1 & \al+\al^2 & 1+\theta^2\al + \theta\al^2 \\ 0 & \theta \al + \theta^2\al^2 & \al + \al^2 \\ 0 & \al+\al^2 & \theta^2\al + \theta \al^2
 \end{bmatrix} \tx(k)
\eeq
We have the following solutions for different initial conditions of the original PFSS under the dynamics of $\tA$ 
 \begin{itemize}
     \item $[0,0,0]^T$ and $[1,0,0]^T$ are fixed points
     \item $[0,0,1]^T,[1,1,0]^T,[0,1,1]^T,[0,1,0]^T,[1,1,1]^T$, $[1,0,1]^T$ have orbits of length $9$
 \end{itemize}
 Since $P(0) = I$, the solutions for the initial conditions of the PFSS under the dynamics of the PFSS are 
 \begin{itemize}
     \item $[0,0,0]^T$ is a fixed point which is in accordance to corollary \ref{cor:FixPtKMA}. 
     \item $[1,0,0]^T$ has a period $3$, which is in accordance to theorem \ref{thm:SolnFloqGenN}
     \item $[0,0,1]^T,[1,1,0]^T,[0,1,1]^T,[0,1,0]^T,[1,1,1]^T$, $[1,0,1]^T$ have orbits of length $\mbox{lcm}(3,9) = 9$ which validates theorem \ref{thm:orbitLCM}.
 \end{itemize}
 The structure of PFSR does look promising to construct FSRs with long orbit lengths. As we see in the example, the PFSR with $3$-states over $\ff_2$ can generate a periodic sequence of period $9$ which is not possible at all with a time invariant FSR where the maximum orbit length is $2^3 = 8$. Also, if the PFSR constructed has a Floquet transform, then the theory developed in this work can be used to find the orbit lengths for various initial conditions. The numerical example provided does validate the theory developed and also serves as an example for a FSS where the Floquet transform involves field extension of degree six. 

\subsubsection{PFSR over extended field}
As discussed earlier, if the Floquet transform exists over an extension of the base field, then we can construct a extended PFSS over that field extension. So, instead of initializing the slave FSR with values from the base field one can initialze the FSR with the values from the extended field too and the Floquet theory helps in computation of the orbit lengths of these initial conditions also. For the equivalent LFSS in equation (\ref{eq:EQLFSS}) all the possible orbit lengths are $1,3$ and $9$. So $T$ for any initial condition can be $1,3$ or $9$ and so from theorem \ref{thm:SolnFloqGenN}, the possible orbits of the PFSS are divisors of $\mbox{lcm}(N,T)$ which are $1,3,9$.


\subsection{Galois PFSR over $\ff_5$}
Consider a $3-$bit Galois PFSR over $\ff_5$ 
Vas shown in figure \ref{fig:GPFSR}. The state transition matrix of the slave FSR of the Galois PFSR configuration at each time is given by \beq
A(k) = \begin{bmatrix} y_1(k) & 1 & 0 \\ y_2(k) & 0 & 1 \\ y_3(k) & 0 & 0\end{bmatrix}
\eeq
The Galois PFSR is non-singular if $y_3(k) \neq 0$ at all $k$. So, we construct a $2-$bit FSR for the master and a $3-$bit FSR for the slave configuration. We assume the FSR to evolve over $\ff_5$. Let the master FSR has the following transition map. 
\[
\begin{bmatrix} y_1(k+1) \\ y_2(k+1) \end{bmatrix} = \begin{bmatrix} 4 & 1 \\ 4 & 0\end{bmatrix} \begin{bmatrix} y_1(k) \\ y_2(k) \end{bmatrix} 
\]
and we initiate the master FSR with the state $[2,3]^T$ which has the following state transition matrices $A(k)$. 
\[
  A(0) = \begin{bmatrix} 2 & 1 & 0 \\ 3 & 0 & 1 \\ 1 & 0 & 0 \end{bmatrix} \hspace{0.2in} A(1) = \begin{bmatrix} 1 & 1 & 0 \\ 3 & 0 & 1 \\ 1 & 0 & 0 \end{bmatrix} \hspace{0.2in} 
  A(2) = \begin{bmatrix} 2 & 1 & 0 \\ 4 & 0 & 1 \\ 1 & 0 & 0 \end{bmatrix} 
 \]
 and the monodromy matrix $\Phi$ is computed
 \[
 \Phi = A(2)A(1)A(0) = \begin{bmatrix} 2&0&2 \\ 2&0&4 \\ 0&1&1\end{bmatrix}
 \]
 Also, the subspace $\kmA$ is computed to be 
 \[
 \kmA = \mbox{span}\{\begin{bmatrix} 0 \\0 \\ 1 \end{bmatrix}\}
 \]
 So any nonzero initial condition not in $\kmA$ will have an orbit length which is a multiple of $3$. We also compute the cube root of $\Phi$ 
 \[
 \tA = \Phi^{1/3} = \begin{bmatrix}1&2&1 \\ 0 & 1 & 2 \\ 0 & 0 & 1 \end{bmatrix}
 \]
 \begin{remark}
Once a matrix $A$ over $\ff$ is converted into Jordan form in some extension $\ff_1$ of $\ff$ with Jordan blocks $J_1,J_2,\dots, J_n$, then the $n-th$ root of $A$ exists if the $n-th$ root exists for each Jordan blocks $J_i$ over some extension of $\ff_1$.   
 \end{remark}
 Since the cube root of $\Phi$ exists, the Galois PFSR allows a Floquet transform and assuming $P(0) = I$, the matrices $P(1)$ and $P(2)$ can be computed. 
 It can be seen that for all non-zero initial conditions, $\tA$ has an orbit of length $5$ which gives $T = 5$.
 \begin{itemize}
     \item $x(0) = \bar{0}$ is a fixed point.  
     \item For $x(0) \notin \kmA$ of the PFSR, the length of orbit of $x(0) = 5 \times 3 = 15$ which is verified to be true. 
     \item For non-zero $x(0) \in \kmA$, the orbit length should divide $15$. It is computed that any non-zero $x(0) \in \kmA$, the orbit length is also 15. 
 \end{itemize}

\section{Conclusion}
A concrete technological dynamical system obtained by compositional FSRs can be modeled as a linear FSS with periodically varying co-efficients called PFSS. The paper develops a theory of PFSS to show that the problem of determining the lengths of periodic orbits and finding initial conditions for orbits of prescribed periods is computationally feasible for such systems with the help of an analogous Floquet theory of these systems. When the Floquet transform exists over an extension field, a concept of equivalent linear systems over the extended field is introduced and the original PFSS itself is then considered with the state space extended over the field. Such as extension of the original PFSS over extended field is meaningful only for finite state systems. In real state space systems complex co-efficients in a time invariant system do not make physical sense. However no such physical limitation exists for finite state systems. Hence an extended system defined over extension field is a generalization of periodic dynamical systems useful for computational purpose. Several examples of computation of orbits with and without Floquet transform are presented along with an example of equivalent system over extension field. The theory presented in this paper is expected to be of interest to a wide community of researchers in application fields such as cryptography, systems biology and control theory.

\bibliographystyle{ieeetr}



\end{document}